\documentclass[11pt]{article} 
\usepackage{fullpage,amsmath,amsthm,amssymb,hyperref,cleveref,graphicx,url,rotating}

\graphicspath{{Fig/}}
  
\title{On a Variational Definition for the Jensen-Shannon Symmetrization of Distances based on the Information Radius}
 
\author{Frank Nielsen\\ Sony Computer Science Laboratories Inc\\ Tokyo, Japan}
 
\newtheorem{Theorem}{Theorem}
\newtheorem{Proposition}{Proposition}

\newtheorem{Definition}{Definition}
\newtheorem{Remark}{Remark}
\newtheorem{Problem}{Problem}
\newtheorem{Example}{Example}

\date{}

\def\Ray{\mathrm{Ray}}
\def\Exp{\mathrm{Exp}}
\def\calQ{\mathcal{Q}}
\def\Wei{\mathrm{Wei}}
\def\calV{\mathcal{V}}
\def\bbR{\mathbb{R}}
\def\Bhat{\mathrm{Bhat}}
\def\TV{\mathrm{TV}}
\def\vJS{\mathrm{vJS}}
\def\JS{\mathrm{JS}}
\def\KL{\mathrm{KL}}
\def\dmu{\mathrm{d}\mu}
\def\tr{\mathrm{tr}}
\def\calF{\mathcal{F}}
\def\calE{\mathcal{E}}
\def\calN{\mathcal{N}}
\def\calQ{\mathcal{Q}}
\def\calC{\mathcal{C}}
\def\calD{\mathcal{D}}
\def\calR{\mathcal{R}}
\def\calP{\mathcal{P}}
\def\calX{\mathcal{X}}
\def\LSE{\mathrm{LSE}}
\def\dmu{\mathrm{d}\mu}
\def\bbP{\mathbb{P}}
\def\minmax{\mathrm{minmax}}

\begin{document}
  
	\maketitle
\begin{abstract}
We generalize the Jensen-Shannon divergence by considering a variational definition with respect to a generic mean extending thereby the notion of Sibson's information radius.
The variational definition applies to any arbitrary distance and yields another way to define a Jensen-Shannon symmetrization of distances.
When the variational optimization is further constrained to belong to prescribed probability measure families, we get relative Jensen-Shannon divergences and symmetrizations which generalize the concept of information projections.
Finally, we discuss applications of these variational Jensen-Shannon divergences and diversity indices to clustering and quantization tasks of probability measures including statistical mixtures.
\end{abstract}


\section{Introduction: Background and motivations}

Let $(\calX,\calF,\mu)$ denote a measure space~\cite{Billingsley-2008} with sample $\calX$, $\sigma$-algebra $\calF$ on the set $\calX$ and positive measure $\mu$ on $(\calF,\mu)$ (e.g., the Lebesgue measure or the counting measure). 
Denote by $\calD=\calD(\calX)$ the set of all densities with full support $\calX$ (Radon-Nikodym derivatives of probability measures with respect to $\mu$):
$$
\calD(\calX) := \left\{p\ :\ \calX\rightarrow\bbR\ :\ p(x)>0\ \mbox{$\mu$-almost everywhere}, \int_\calX p(x)\dmu(x)=1\right\}.
$$

The {\em Jensen-Shannon divergence}~\cite{JS-1991} (JSD) between two densities $p$ and $q$ of $\calD$ is defined by:

\begin{equation}\label{eq:jsd}
D_\JS[p,q] :=  \frac{1}{2}\left(D_\KL\left[p:\frac{p+q}{2}\right]+D_\KL\left[q:\frac{p+q}{2}\right]\right),
\end{equation}
where $D_\KL$ denotes the {\em Kullback-Leibler divergence}~\cite{Kullback-1997,CoverThomasIT-2012} (KLD):

\begin{equation}\label{eq:kld}
D_\KL[p:q]:=  \int_\calX p(x)\log\left(\frac{p(x)}{q(x)}\right)\dmu(x).
\end{equation}

The JSD belongs to the class of {\em $f$-divergences}~\cite{fdivMorimoto-1963,Csiszar-1964,fdiv-AliSilvey-1966}, the {\em invariant decomposable divergences} of information geometry (see~\cite{IG-2016}, pp. 52-57). 
Although the KLD is asymmetric (i.e., $D_\KL[p:q]\not= D_\KL[q:p]$), the JSD is symmetric (i.e., $D_\JS[p,q] =D_\JS[q,p]$).
The notation `:' is used as a parameter separator to indicate that the parameters are not permutation invariant and that the order of parameters is important. 

The {\em $2$-point JSD} of Eq.~\ref{eq:kld} can be extended to a weighted set of $n$ densities 
$\calP:=\{(w_1,p_1), \ldots, (w_n,p_n)\}$ (with positive $w_i$'s normalized to sum up to unity, i.e., $\sum_{i=1}^n w_i=1$) thus providing a {\em diversity index}, i.e., a {\em $n$-point JSD} for $\calP$:
\begin{equation}\label{eq:jsdiv}
D_\JS(\calP) :=   \sum_{i=1}^n w_i D_\KL\left[p_i:\bar{p}\right],
\end{equation}
where $\bar{p}:=\sum_{i=1}^n w_ip_i$ denotes the statistical mixture~\cite{Mixtures-2004} of the densities of $\calP$.

The KLD is also called the {\em relative entropy} since it can be expressed as  the difference between the {\em cross entropy} $h[p:q]$ and the {\em entropy} $h[p]$:
\begin{equation}
D_\KL[p:q] = h[p:q]-h[p],
\end{equation}
 with the cross-entropy and entropy defined respectively by
\begin{eqnarray}
h[p:q]&:=&-\int_\calX p(x)\log q(x)\dmu(x),\\
h[p]&:=&-\int_\calX p(x)\log p(x)\dmu(x).
\end{eqnarray} 
Since $h[p]=h[p:p]$, we may say that the entropy is the self-cross-entropy.

When $\mu$ is the Lebesgue measure, the Shannon entropy is also called the {\em differential entropy}~\cite{CoverThomasIT-2012}.
Although the discrete entropy $H[p]=-\sum_i p_i\log p_i$ (i.e., entropy with respect to the counting measure) is always positive and bounded by $\log|\calX|$, the differential entropy may be negative (e.g., entropy of a Gaussian distribution with small variance).

The Jensen-Shannon diversity index of Eq.~\ref{eq:jsdiv} can be rewritten as:
\begin{equation}\label{eq:hjsd}
D_\JS[p,q] =  h[\bar{p}] - \sum_{i=1}^n w_i h[p_i].
\end{equation}
The JSD representation of Eq.~\ref{eq:hjsd} is a {\em Jensen divergence}~\cite{BR-2011} for the strictly convex negentropy $F(p)=-h(p)$ since the entropy function $h[.]$ is strictly concave, hence its name {\em Jensen-Shannon divergence}.

Since $\frac{p_i(x)}{\bar{p}(x)}\leq \frac{p_i(x)}{w_ip_i(x)}=\frac{1}{w_i}$, it can be shown that the Jensen-Shannon diversity index is  upper bounded by $H(w):=-\sum_{i=1}^n w_i\log w_i$, the {\em discrete Shannon entropy}.
In particular, the JSD is bounded by $\log 2$ although the KLD is unbounded and may even be equal to $+\infty$ when the definite integral diverges (e.g., KLD between the standard Cauchy distribution and the standard Gaussian distribution).
Another nice property of the JSD is that its square root yields a {\em metric distance}~\cite{JSmetric-2003,JSmetric-2004}.
This property further holds for the quantum JSD~\cite{QuantumJSD-2021}. 
Recently, the JSD has gained interest in machine learning.
See for example, the Generative Adversarial Networks~\cite{GAN-2014} (GANs) in deep learning~\cite{DL-2016} where it was proven that minimizing the GAN objective function by adversarial training is equivalent to minimizing a JSD.

The Jensen-Shannon divergence can be skewed using two scalars $\alpha,\beta \in (0,1)$ as follows:
\begin{eqnarray}
D_{\JS,\alpha,\beta}[p:q]&:=&(1-\beta)D_\KL[p:m_\alpha]+\beta D_\KL[q:m_\alpha],\\
&=& h[m_\beta:m_\alpha]-\left((1-\beta)h[p]+\beta h[q]\right),
\end{eqnarray}
where $m_\alpha:=(1-\alpha)p+\alpha q$ and $m_\beta:=(1-\beta)p+\beta q$, and $h[p:q]$ denotes the cross-entropy:
\begin{equation}
h[p:q] := -\int p(x)\log q(x)\dmu(x).
\end{equation}
Thus when $\beta=\alpha$, we have $D_{\JS,\alpha}[p,q]=D_{\JS,\alpha,\alpha}[p,q]=h[m_\alpha]-((1-\alpha)h[p]+\alpha h[q])$ since the self-cross entropy corresponds to the entropy: $h[m_\alpha:m_\alpha]=h[\alpha]$.

A $f$-divergence~\cite{Csiszar-2008} is defined for a convex generator $f$ strictly convex at $1$ with $f(1)=0$ by
\begin{equation}
I_f[p:q]=\int p(x) f\left(\frac{q(x)}{p(x)}\right)\dmu(x).
\end{equation}

The $D_{\JS,\alpha,\beta}$ divergence is a $f$-divergence for the generator:
\begin{equation}
f_{\JS,\alpha,\beta}(u)=-\left((1-\beta)\log\left(\alpha u+(1-\alpha)\right)+\beta u\log\left(\frac{1-\alpha}{u}+\alpha\right)\right).
\end{equation}
We check that the generator $f_{\JS,\alpha,\beta}$ is strictly convex since for any $a\in(0,1)$ and $b\in(0,1)$, we have
\begin{equation}
f_{\JS,\alpha,\beta}''(u)= \frac{a^2(1-b)u+(a-1)^2b}{a^2u^3+2a(1-a)u^2+(a-1)^2u}>0,
\end{equation}
when $u>0$.
We have $D_{\JS,\alpha,\beta}[p:q]=I_{f_{\JS,\alpha,\beta}}[p:q]$.


The Jensen-Shannon {\em principle} of taking the average of the (Kullback-Leibler) divergences between the source parameters to the mid-parameter can be applied to other distances. For example, the Jensen-Bregman divergence is a {\em Jensen-Shannon symmetrization} of the Bregman divergence $B_F$~\cite{BR-2011}:
\begin{equation}
B_F^\JS(\theta_1:\theta_2) := \frac{1}{2}\left( B_F\left(\theta_1:\frac{\theta_1+\theta_2}{2}\right)+B_F\left(\theta_2:\frac{\theta_1+\theta_2}{2}\right) \right).
\end{equation}
The Jensen-Bregman divergence $B_F^\JS$ can also be written as an equivalent Jensen divergence $J_F$:
\begin{equation}
B_F^\JS(\theta_1:\theta_2) = J_F(\theta_1:\theta_2) := \frac{F(\theta_1)+F(\theta_2)}{2}-F\left(\frac{\theta_1+\theta_2}{2}\right),
\end{equation}
where $F$ is a strictly convex function ensuring $ J_F(\theta_1:\theta_2)\geq 0$ with equality iff $\theta_1=\theta_2$.

Because of its use in various fields of information sciences~\cite{JSApp-2009}, various {\em generalizations} of the JSD have been proposed:
These generalizations are 
either based on Eq.~\ref{eq:jsd}~\cite{JSsym-2019} or are based on Eq.~\ref{eq:hjsd}~\cite{nielsen2010family,JensenComparative-2017,JScentroid-2020}.
For example,  the (arithmetic) mixture $\bar{p}=\sum_i w_i p_i$ in Eq.~\ref{eq:jsd} was replaced by an abstract statistical mixture with respect to a generic mean $M$ in~\cite{JSsym-2019} (e.g., geometric mixture induced by the geometric mean), and the two KLDS defining the JSD in Eq.~\ref{eq:jsd} was further averaged using another abstract mean $N$, thus yielding the
 following generic {\em $(M,N)$-Jensen-Shannon divergence}~\cite{JSsym-2019} (abbreviated as $(M,N)$-JSD):

\begin{equation}\label{eq:mnjsd}
D_\JS^{M,N}[p:q] :=   N\left( D_\KL\left[p:(pq)^M_{\frac{1}{2}}\right], D_\KL\left[q:(pq)^M_{\frac{1}{2}}\right]\right),
\end{equation}
where $(pq)^M_{\alpha}$ denotes the statistical weighted $M$-mixture:

\begin{equation}
(pq)^M_{\alpha}:= \frac{M_\alpha(p(x),q(x))}{\int_\calX M_\alpha(p(x),q(x))\dmu(x)}.
\end{equation}
Notice that when $M=N=A$ (the arithmetic mean), Eq.~\ref{eq:mnjsd} of the $(A,A)$-JSD reduces to the ordinary JSD of Eq.~\ref{eq:jsd}.
When the means $M$ and $N$ are symmetric, the $(M,N)$-JSD is symmetric.

In general, a {\em weighted mean} $M_\alpha(a,b)$ for any $\alpha\in [0,1]$ shall satisfy the {\em in-betweeness property}:
\begin{equation}
\min\{a,b\}\leq M_\alpha(a,b)\leq \max\{a,b\}.
\end{equation}
A weighted mean (also called barycenter) can be built from a non-weighted mean $M(a,b)$ (i.e., $\alpha=\frac{1}{2}$) using the dyadic expansion of the weight $\alpha$ as detailed in~\cite{WeightedMean-2018}. 

The three {\em Pythagorean means} defined for positive scalars $a>0$ and $b>0$ are classic examples of means: 
\begin{itemize}
\item The {\em arithmetic mean} $A(a,b)=\frac{a+b}{2}$, 
\item the {\em geometric mean} $G(a,b)=\sqrt{ab}$, and
\item the {\em harmonic mean} $H(a,b)=\frac{2ab}{a+b}$.
\end{itemize}

These Pythagorean means may be interpreted as special instances of another parametric family of means: 
The {\em power means} 
\begin{equation}
P_\alpha(a,b) := \left(\frac{a^\alpha+b^\alpha}{2}\right)^{\frac{1}{\alpha}},
\end{equation}
defined for $\alpha\in\bbR\backslash\{0\}$ (also called H\"older means).
The power means can be extended to the full range $\alpha\in\bbR$ by using the property that 
$\lim_{\alpha\rightarrow 0} P_\alpha(a,b)=G(a,b)$. 
The power means are homogeneous means: $P_\alpha(\lambda a,\lambda b)=\lambda P_\alpha(a,b)$ for any $\lambda>0$.
We refer to the handbook of means~\cite{Bullen-2013}  to get definitions and principles of other means beyond these power means.

Choosing the abstract mean $M$ in accordance with the family of the densities allows one to obtain closed-form formula for the $(M,N)$-JSDs which rely on definite integral calculations.
For example, the JSD between two Gaussian densities does not admit a closed-form formula because of the log-sum integral, but the $(G,N)$-JSD admits a closed-form formula when using {\em geometric statistical mixtures} (i.e., when $M=G$).
As an application of these generalized JSDs, Deasy et al.~\cite{VIGJSD-2020} used the skewed geometric JSD (namely, the $(G_\alpha,A_{1-\alpha})$-JSD for $\alpha\in(0,1)$) which admits  a closed-form formula between normal densities~\cite{JSsym-2019}, and showed how regularizing an optimization task with this G-JSD divergence improved reconstruction and generation of Variational AutoEncoders (VAEs).

More generally, instead of using the KLD, one can also use any arbitrary distance $D$ to define its {\em JS-symmetrization} as follows:
\begin{equation}\label{eq:genmnjsd}
D^\JS_{M,N}[p:q] :=   N\left( D\left[p:(pq)^M_{\frac{1}{2}}\right], D\left[q:(pq)^M_{\frac{1}{2}}\right]\right).
\end{equation}
These symmetrizations may further be skewed by using $M_\alpha$ and/or $N_\beta$ for $\alpha\in (0,1)$ and $\beta\in(0,1)$ yielding the definition~\cite{JSsym-2019}:
\begin{equation}\label{eq:skewgenmnjsd}
D^\JS_{M_\alpha,N_\beta}[p:q] :=   N_\beta\left( D\left[p:(pq)^M_{\alpha}\right], D\left[q:(pq)^M_{\alpha}\right]\right).
\end{equation}
With these notations, the ordinary JSD is $D_\JS={D_\KL}^\JS_{A,A}$, the $(A,A)$ JS-symmetrization of the KLD with respect to the arithmetic means $M=A$ and $N=A$.

In this work, we consider symmetrizing an arbitrary distance $D$ (including the KLD) generalizing the Jensen-Shannon divergence by using a
 {\em variational formula} for the JSD. 
Namely, we observe that the Jensen-Shannon divergence can also be defined as the following minimization problem:

\begin{equation}\label{eq:varjs}
D_\JS[p,q] := \min_{c\in\calD} \frac{1}{2}\left( D_\KL[p:c]+D_\KL[q:c] \right),
\end{equation}
since the optimal density $c$ is proven unique using the calculus of variation~\cite{Sibson-1969,Amari-2007,IG-2014} and corresponds to the mid density $\frac{p+q}{2}$, a statistical (arithmetic) mixture.

\begin{proof}
Let $S(c)=D_\KL[p:c]+D_\KL[q:c]\geq 0$.
We use the method of the Lagrange multipliers for the constrained optimization problem $\min_c S(c)$ such that $\int c(x)\dmu(x)=1$.
Let us minimize $S(c)+\lambda\left(\int c(x)\dmu(x)-1\right)$.
The density $c$ realizing the minimum $S(c)$ satisfies the Euler-Lagrange equation $\frac{\partial L}{\partial c}=0$
where $L(c):=p\log\frac{p}{c}+q\log\frac{q}{c}+\lambda c$ is the Lagrangian.
That is, $-\frac{p}{c}-\frac{q}{c}+\lambda=0$, or equivalently $c=\frac{1}{\lambda}(p+q)$.
Parameter $\lambda$ is then evaluated from the constraint $\int_\calX c(x)\dmu(x)=1$:
We get $\lambda=2$ since $\int_\calX (p(x)+q(x))\dmu(x)=2$.
Therefore, we find that $c(x)=\frac{p(x)+q(x)}{2}$, the mid density of $p(x)$ and $q(x)$.
\end{proof}

The paper is organized as follows:
In~\S\ref{sec:ir}, we recall the rationale and definitions of the R\'enyi $\alpha$-entropy and the R\'enyi $\alpha$-divergence~\cite{Renyi-1961}, and explain the information radius of Sibson~\cite{Sibson-1969} which includes as a special case the ordinary Jensen-Shannon divergence and which can be interpreted as generalized skew Bhattacharyya distances. It is noteworthy to point out that Sibson's work (1969) includes as a particular case of the information radius a definition of the JSD, prior to the reference paper of Lin~\cite{JS-1991} (1991).
In~\S\ref{sec:irgen}, we present the JS-symmetrization variational definition based on a generalization of the information radius with a  generic mean (Definition~\ref{def:genJSD} and Definition~\ref{def:JSsym}).
In~\S\ref{sec:iref}, we constrain the mixture density to belong to a prescribed class of (parametric) probability densities like an exponential family~\cite{EF-2014}, and get relative information radius generalizing information radius and related to the concept of information projections.
Our definition~\ref{def:relJSsym} generalizes the (relative) {\em normal information radius} of Sibson~\cite{Sibson-1969} who considered the multivariate normal family.
As an application of these relative variational JSDs, we consider clustering and quantization of probability densities in \S\ref{sec:clustering}.
Finally, we conclude by summarizing our contributions and discussing related works in~\S\ref{sec:concl}.

\section{R\'enyi entropy and divergence, and Sibson information radius}\label{sec:ir}

R\'enyi~\cite{Renyi-1961} investigated a generalization of the four axioms of Fadeev~\cite{Fadeev-1957} yielding to the unique Shannon entropy~\cite{Csiszar-2008}.  
In doing so, R\'enyi replaced the ordinary weighted arithmetic mean by a more general class of averaging schemes.
Namely, R\'enyi considered the {\em weighted quasi-arithmetic means}~\cite{Kolmogorov-1930}.
A  weighted quasi-arithmetic mean can be induced by a strictly monotonous and continuous function $g$ as follows:
\begin{equation}
M_g(x_1,\ldots,x_n;w_1,\ldots,w_n) := g^{-1}\left(\sum_{i=1}^n w_i g(x_i)\right),
\end{equation}
where the $x_i$'s and the $w_i$'s are positive (the weights are normalized so that $\sum_{i=1}^n w_i=1$).
Since $M_g=M_{-g}$, we may assume without loss of generality that $g$ is a strictly increasing and continuous function.
The quasi-arithmetic means were investigated independently by Kolmogorov~\cite{Kolmogorov-1930}, Nagumo~\cite{Nagumo-1930}, and de Finetti~\cite{Finetti-1931}.

For example, the {power means} $P_\alpha(a,b)=\left(\frac{a^\alpha+b^\alpha}{2}\right)^{\frac{1}{\alpha}}$ 
introduced earlier are quasi-arithmetic means for the generator $g_\alpha^P(u):=u^\alpha$: 
\begin{equation}
P_\alpha(a,b)=M_{g_\alpha^P}\left(a,b;\frac{1}{2},\frac{1}{2}\right).
\end{equation}

R\'enyi proved that among the class of weighted quasi-arithmetic means, only the means induced by the family of functions
\begin{eqnarray}
g_\alpha(u)&:=&2^{(\alpha-1)u},\\
g_\alpha^{-1}(v)&:=& \frac{1}{\alpha-1}\log_2 v,
\end{eqnarray}
for $\alpha>0$ and $\alpha\not=1$ yield a proper generalization of Shannon entropy called nowadays the R\'enyi $\alpha$-entropy.
The {\em R\'enyi $\alpha$-mean} is 
\begin{eqnarray}
M_\alpha^R(x_1,\ldots,x_n;w_1,\ldots,w_n) &=& M_{g_\alpha}\left(x_1,\ldots,x_n;w_1,\ldots,w_n\right),\\
&=&  \frac{1}{\alpha-1}\log_2 \left(\sum_{i=1}^n w_i 2^{(\alpha-1)x_i}\right).
\end{eqnarray}
The R\'enyi $\alpha$-means $M_\alpha^R$ are not power means: They are not homogeneous means~\cite{Amari-2007}.
Let $M_\alpha^R(p,q)=M_\alpha^R\left(p,q;\frac{1}{2},\frac{1}{2}\right)=\frac{1}{\alpha-1}\log_2 \frac{2^{(\alpha-1)p}+2^{(\alpha-1)q}}{2}$. Then we have
$\lim_{\alpha\rightarrow\infty} M_\alpha^R(p,q)=\max\{p,q\}$ and $\lim_{\alpha\rightarrow\infty} M_\alpha^R(p,q)=A(p,q)=\frac{p+q}{2}$.
Indeed, we have
\begin{eqnarray*}
M_\alpha^R(p,q) &=& \frac{1}{\alpha-1}\log_2 \frac{2^{(\alpha-1)p}+2^{(\alpha-1)q}}{2},\\
&=&\frac{1}{\alpha-1}\log_2 \frac{e^{(\alpha-1)p\log 2}+e^{(\alpha-1)q\log 2}}{2},\\
&\approx_{\alpha\rightarrow 0}& \frac{1}{\alpha-1}\log_2 \left(1+(\alpha-1)\frac{p+q}{2}\log 2\right),\\
&\approx_{\alpha\rightarrow 0}& \frac{1}{\alpha-1} \frac{1}{\log 2}{(\alpha-1)\frac{p+q}{2}\log 2},\\
&\approx_{\alpha\rightarrow 0}&  \frac{p+q}{2}= A(p,q),
\end{eqnarray*}
using the following first-order approximations: $e^x\approx_{x\rightarrow 0}=1+x$ and $\log(1+x)\approx_{x\rightarrow 0}=x$.

To get an intuition of the R\'enyi entropy, we may consider generalized entropies derived from quasi-arithmetic means as follows:
\begin{equation}
h_g[p]:=-M_g(\log_2 p_1,\ldots, \log_2 p_n;p_1,\ldots,p_n).
\end{equation} 
When $g(u)=u$, we recover Shannon entropy. When $g_2(u)=2^u$, we get $h_{g_2}[p]=-\log \sum_i p_i^2$, called the collision entropy since $-\log\mathrm{Pr}[X_1=X_2]=h_{g_2}[p]$ when $X_1$ and $X_2$ are independent and identically distributed random variables with
 $X_1\sim p$ and $X_2\sim p$.
When $g(u)=g_\alpha(u)=2^{(\alpha-1)u}$, we get 
\begin{eqnarray}
h_{g_\alpha}[p] &=& -\frac{1}{\alpha-1}\log_2 \left(\sum_i p_i 2^{(\alpha-1)\log_2 p_i} \right),\\
&=&  \frac{1}{1-\alpha}\log_2 \sum_i p_i p_i^{\alpha-1}=\frac{1}{1-\alpha}\log_2 \sum_i p_i^{\alpha}.
\end{eqnarray}

The R\'enyi $\alpha$-entropy~\cite{Renyi-1961} are defined by:
\begin{equation}
h_\alpha^R[p] := \frac{1}{1-\alpha}\log\left( \int_\calX p^\alpha(x)\dmu(x)\right), \quad \alpha\in (0,1)\cup(1,\infty).
\end{equation}
In the limit case $\alpha\rightarrow 1$, the R\'enyi $\alpha$-entropy converges to Shannon entropy:
$\lim_{\alpha\rightarrow 1} h_\alpha^R[p]=h[p]=-\int_\calX p(x)\log p(x)\dmu(x)$.
R\'enyi $\alpha$-entropies are {\em non-increasing} with respect to increasing $\alpha$: 
$h_\alpha^R[p]\geq h_{\alpha'}^R[p]$ for $\alpha<\alpha'$.
In the discrete case (i.e., counting measure $\mu$ on a finite alphabet $\calX$), we can further define $h_0[p]=\log |\calX|$ for $\alpha=0$ (also called max-entropy or Hartley entropy).
The R\'enyi $+\infty$-entropy 
$$
h_{+\infty}[p]=-\log\max_{x\in\calX} p(x)
$$ 
is also called the {\em min-entropy} since the sequence $h_\alpha$ is non-increasing with respect to increasing $\alpha$.

Similarly, R\'enyi obtained the $\alpha$-divergences for $\alpha>0$ and $\alpha\not=1$ (originally called information gain of order $\alpha$):
\begin{equation}
D_\alpha^R[p:q] := \frac{1}{\alpha-1}\log_2\left( \int_\calX p(x)^\alpha q(x)^{1-\alpha}\dmu(x)\right),
\end{equation}
generalizing the Kullback-Leibler divergence
since $\lim_{\alpha\rightarrow 1} D_\alpha^R[p:q]=D_\KL[p:q]$.
R\'enyi $\alpha$-divergences are {\em non-decreasing} with respect to increasing $\alpha$~\cite{RenyiDiv-2014}:
$D_\alpha^R[p:q]\leq D_{\alpha'}^R[p:q]$ for $\alpha'\geq\alpha$.

Sibson\footnote{Robin Sibson (1944-2017) is also renown for inventing the natural neighbour interpolation~\cite{Sibson-1981}.}~\cite{Sibson-1969} considered both the R\'enyi $\alpha$-divergence~\cite{Renyi-1961} $D^R_\alpha$ 
and the R\'enyi $\alpha$-weighted mean $M^R_\alpha:=M_{g_\alpha}$
to define the {\em information radius} $R_\alpha$ of order $\alpha$ of a weighted set  $\calP=\{(w_i,p_i)\}_{i=1}^n$ of densities $p_i$'s as the following minimization problem:
\begin{equation}\label{eq:ireq}
R_\alpha(\calP) := \min_{c\in\calD} R_\alpha(\calP,c),
\end{equation}
where
\begin{equation}
R_\alpha(\calP,c) := M^R_\alpha\left(D^R_\alpha[p_1:c],\ldots, D^R_\alpha[p_n:c]; w_1,\ldots,w_n\right),
\end{equation}
and
\begin{eqnarray}\label{eq:renyiweight}
M^R_\alpha(x_1,\ldots,x_n;w_1,\ldots,w_n )&=& \frac{1}{\alpha-1}\log_2 \sum_{i=1}^n w_i 2^{(\alpha-1)x_i},\\
&=& \frac{1}{\alpha-1} \LSE\left((\alpha-1)x_1\log 2+\log w_1,\ldots, (\alpha-1)x_i\log 2+\log w_i\right).
\end{eqnarray}
Function $\LSE(a_1,\ldots,a_n):=\log\left(\sum_{i=1}^n e^{a_i}\right)$ denotes the log-sum-exp (convex) function~\cite{ConvexOptim-2004,KLLSE-2016}.

Notice that $2^{(\alpha-1)D_\alpha^R[p:q]}=\int_\calX p(x)^\alpha q(x)^{1-\alpha}\dmu(x)$, the {\em Bhattacharyya $\alpha$-coefficient}~\cite{BR-2011} (also called Chernoff $\alpha$-coefficient~\cite{Chernoff-2011,Chernoff-2013}):
\begin{equation}
C_{\Bhat,\alpha}[p:q]:=\int_\calX p(x)^\alpha q(x)^{1-\alpha}\dmu(x).
\end{equation}
Thus we have
\begin{equation}
R_\alpha(\calP,c)=\frac{1}{\alpha-1}\log_2 \left(\sum w_i C_{\Bhat,\alpha}[p_i:c]\right).
\end{equation}
The ordinary Bhattacharyya coefficient is obtained for $\alpha=\frac{1}{2}$: $C_{\Bhat}[p:q]:=\int_\calX \sqrt{p(x)} \sqrt{q(x)}\dmu(x)$.

Sibson~\cite{Sibson-1969} also considered the limit case $\alpha\rightarrow\infty$ when defining the information radius:
\begin{equation}
D_\infty^R[p:q]:=\log_2 \sup_{x\in\calX} \frac{p(x)}{q(x)}.
\end{equation}

Sibson reported the following theorem in his information radius study~\cite{Sibson-1969}:

\begin{Theorem}[Theorem 2.2 and Corollary 2.3 of~\cite{Sibson-1969}]\label{thm:Sibson}
The optimal density $c^*_\alpha=\arg\min_{c\in\calD} R_\alpha(\calP,c)$ is unique and we have:
$$
\begin{array}{lll}
c^*_1(x)=\sum_i w_ip_i(x), & R_1(\calP)=R_1(\calP,c_1^*)=\int_\calX \sum_i w_ip_i\log_2\frac{p_i}{\sum_j w_jp_j(x)}\dmu(x), &   \\
c^*_\alpha(x)=\frac{(\sum_i w_ip_i(x)^\alpha)^{\frac{1}{\alpha}} }{\int_\calX (\sum_i w_ip_i(x)^\alpha)^{\frac{1}{\alpha}}\dmu(x)}, &
R_\alpha(\calP)=R_\alpha(\calP,c_\alpha^*)=
\frac{1}{\alpha-1}\log_2 \left(\int_\calX (\sum_i w_i p_i(x)^\alpha)^{\frac{1}{\alpha}}\dmu(x)\right)^{\alpha},  
 & \\
 & \alpha\in(0,1)\cup(1,\infty)\\
c_\infty^*(x)= \frac{\max_i p_i(x)}{\int_\calX (\max_i p_i(x))\dmu(x)}, & R_\infty(\calP)=R_\infty(\calP,c_\infty^*)
=\log_2 \int_\calX \left(\max_i p_i(x)\right)\dmu(x),& 
\end{array}
$$
\end{Theorem}

Observe that $R_\infty(\calP)$ does not depend on the (positive) weights.

The proof follows from the following decomposition of the information radius:
\begin{Proposition}\label{prop:decomp}
We have:
\begin{equation}
R_\alpha(\calP,c) - R_\alpha(\calP,c^*_\alpha) = D_\alpha^R(c^*_\alpha,c)\geq 0.
\end{equation}
\end{Proposition}

Since the proof is omitted in~\cite{Sibson-1969}, we report it here:

\begin{proof}
Let $\Delta(c,c^*_\alpha):=R_\alpha(\calP,c) - R_\alpha(\calP,c^*_\alpha)$.
We handle the three cases depending on the $\alpha$ values:

\begin{itemize}
\item Case $\alpha\in(0,1)\cup(1,\infty)$:
Let $P_\alpha(\calP)(x):=\left(\sum_i w_i p_i(x)^\alpha\right)^{\frac{1}{\alpha}}$.
We have $(c_\alpha^*(x))^\alpha = \frac{\sum_i  w_i p_i(x)^\alpha}{\left(\int P_\alpha(\calP)(x)\dmu(x)\right)^\alpha}$.
We get
\begin{eqnarray}
\Delta(c,c^*_\alpha) &=& \frac{1}{\alpha-1}\log_2\left(\sum_i w_i\int p_i(x)^\alpha c(x)^{1-\alpha}\dmu(x)\right)
- \frac{1}{\alpha-1}\log_2\left( \int P_\alpha(\calP)(x)\dmu(x)\right)^\alpha,\\
&=&\frac{1}{\alpha-1}\log_2 \frac{\sum_i w_i\int p_i(x)^\alpha c(x)^{1-\alpha}\dmu}{(\int P_\alpha(\calP)(x)\dmu(x))^\alpha},\\
&=&\frac{1}{\alpha-1}\log_2 \frac{\int  (\sum_i w_ip_i(x)^\alpha) c(x)^{1-\alpha}}{(\int P_\alpha(\calP)(x)\dmu(x))^\alpha}\dmu(x),\\
&=&\frac{1}{\alpha-1}\log_2 \int (c_\alpha^*(x))^\alpha c(x)^{1-\alpha}\dmu(x),\\
&:=& D_\alpha^R(c^*_\alpha,c).
\end{eqnarray}

\item Case $\alpha=1$: We have $\Delta(c,c^*_1):=R_1(\calP,c) - R_1(\calP,c^*_1)$ with $c^*_1=\sum_i w_ip_i$.
Since $R_1(\calP,c)=\sum_i w_i D_\KL[p_i:c]$, we have 
\begin{eqnarray}
R_1(\calP,c) &=&\sum_i w_i h[p_i:c]- w_i h[p_i],\\
&=& h[\sum_i w_ip_i:c] - \sum_i w_i h[p_i],\\
&=&  h[c^*_1:c] - \sum_i w_i h[p_i].
\end{eqnarray}
It follows that
\begin{eqnarray}
\Delta(c,c^*_1)&=&h[c^*_1:c] - \sum_i w_i h[p_i] - \left(h[c^*_1:c_1^*] - \sum_i w_i h[p_i]\right),\\
&=& h^[c^*_1:c]-h[c^*_1],\\
&=& D_\KL[c^*_1:c]=D_1^R[c^*_1:c].
\end{eqnarray}

\item Case   $\alpha=\infty$: we have $c_\infty^*=\frac{\max_i p_i(x)}{\int (\max_i p_i(x))\dmu(x)}$,
  $R_\infty(\calP,c_\infty^*)=\log_2 \int (\max_i p_i(x))\dmu(x)$, and $D_\infty^R[p:q]=\log_2 \sup_x \frac{p(x)}{q(x)}$.
	We have $R_\infty(\calP,c)=\log_2 \sup_x \frac{p_i(x)}{c(x)}$
	Thus $\Delta(c,c^*_\alpha):=R_\infty(\calP,c) - R_\infty(\calP,c^*_\infty)=\log_2 \sup_x \frac{c_\infty^*(x)}{c(x)}=D_\infty^R[c^*_\infty:c]$.
	\end{itemize}
\end{proof}

It follows that
$$
\min_c R_\alpha(\calP,c)  = \min_c R_\alpha(\calP,c^*_\alpha)+D_\alpha^R(c^*_\alpha,c)\equiv \min_c D_\alpha^R(c^*_\alpha,c)\geq 0.
$$
Thus we have $c=c^*_\alpha$ since $D_\alpha^R(c^*_\alpha,c)$ is minimized for $c=c^*_\alpha$.

Notice that $c_\infty^*(x)= \frac{\max\{p_1(x),\ldots, p_n(x)\}}{\int_\calX (\max_i p_i(x))\dmu(x)}$ is the upper envelope of the densities $p_i(x)$'s normalized to be a density.
Provided that the densities $p_i$'s intersect pairwise in at most $s$ locations (i.e., $|\{p_i(x)\cap p_j(x)\}|\leq s$ for $i\not=j$), we can compute efficiently this upper envelope using an output-sensitive algorithm~\cite{Nielsen-1998} of computational geometry.

When the point set is $\calP=\left\{\left(\frac{1}{2},p\right),\left(\frac{1}{2},q\right)\right\}$ with $w_1=w_2=\frac{1}{2}$, the information radius defines a ($2$-point) symmetric distance as follows:

$$
\begin{array}{ll}
R_1(p,q)= \frac{1}{2} \int_\calX p(x)\log_2\frac{2p}{p(x)+q(x)}\dmu(x) + \frac{1}{2} \int_\calX q(x)\log_2\frac{2q(x)}{p(x)+q(x)}\dmu(x), &   \alpha=1 \\
R_\alpha(p,q)=\frac{\alpha}{\alpha-1}\log_2 \int_\calX \left(\frac{ p(x)^\alpha+ q(x)^\alpha}{2}\right)^{\frac{1}{\alpha}} \dmu(x) 
=\frac{\alpha}{\alpha-1}\log_2 \int_\calX P_\alpha(p(x),q(x)) \dmu(x),
 & \alpha\in(0,1)\cup(1,\infty)\\
R_\infty(p_1,p_2)=\log_2 \int_\calX \max\{p(x),q(x)\} \dmu(x),& \alpha=\infty.
\end{array}
$$

This family of symmetric divergences may be called the Sibson's $\alpha$-divergences, and the Jensen-Shannon divergence is interpreted
 as a limit case when $\alpha\rightarrow 1$.
Notice that since we have $\lim_{\alpha\rightarrow\infty} P_\alpha(p,q)=\max\{p,q\}$ and $\lim_{\alpha\rightarrow\infty}\frac{\alpha}{\alpha-1}=1$, we have $\lim_{\alpha\rightarrow\infty} R_\alpha(p,q)=R_\infty(p_1,p_2)$.
Notice that for $\alpha=1$, the integral and logarithm operations are swapped compared to $R_\alpha$ for $\alpha\in(0,1)\cup(1,\infty)$.

\begin{Theorem}\label{thm:cfirfracinteger}
When $\alpha=\frac{1}{k}$ for an integer $k\geq 2$, the Sibson $\alpha$-divergences between two densities $p_{\theta_1}$ and $p_{\theta_2}$ of an exponential family $\{p_\theta\ :\ \theta\in\Theta\}$ with cumulant function $F(\theta)$ is available in closed form:
$$
R_\alpha(p_{\theta_1},p_{\theta_2})=-\frac{1}{k-1}\log_2 \left( \frac{1}{2^k} \sum_{i=0}^k \binom{k}{i}  \exp\left(F\left(\frac{i}{k}\theta_1+\left(1-\frac{i}{k}\right)\theta_2\right)-\left(\frac{i}{k}F(\theta_1)+\left(1-\frac{i}{k}\right)F(\theta_2)\right)\right) \right).
$$
\end{Theorem}

\begin{proof}
Let $p=p_{\theta_1}$ and $q=p_{\theta_2}$ be two densities of an exponential family~\cite{EF-2014} with  cumulant function $F(\theta)$ and natural  parameter space $\Theta$. 
Without loss of generality, we may consider a natural exponential family~\cite{EF-2014} with densities written canonically as $p_\theta(x)=\exp(x^\top\theta-F(\theta))$ for $\theta\in\Theta$.
It can be shown that the cumulant function $F(\theta)=\log \int_\calX \exp(x^\top\theta) \dmu(x)$ is strictly convex and analytic on the open convex natural parameter space $\Theta$~\cite{EF-2014}.

When $\alpha=\frac{1}{2}$ (i.e., $k=2$), we have:
\begin{eqnarray}
R_{\frac{1}{2}}(p,q) &=& -\log_2 \int_\calX \left(\frac{ \sqrt{p(x)}+ \sqrt{q(x)}}{2}\right)^{2} \dmu(x),\\ 
&=& -\log_2 \left(\frac{1}{2}+ \frac{1}{2}\int_\calX \sqrt{p(x)}\sqrt{q(x)}\dmu(x)\right),\\
&=& -\log_2 \left(\frac{1}{2}+ \frac{1}{2} C_{\Bhat}[p:q]\right) \geq 0,
\end{eqnarray}
where $C_{\Bhat}[p:q]:=\int_\calX \sqrt{p(x)}\sqrt{q(x)}\dmu(x)$  is the Bhattacharyya coefficient (with $0 \leq C_{\Bhat}[p:q]\leq 1$). Using Theorem~3 of~\cite{BR-2011}, we have
$$
C_{\Bhat}[p_{\theta_1},p_{\theta_2}]=\exp\left( F\left(\frac{\theta_p+\theta_q}{2}\right) - \frac{F(\theta_p)+F(\theta_q)}{2}\right),
$$
so that we get the following closed-form formula:
$$
R_{\frac{1}{2}}(p_{\theta_1},p_{\theta_2}) = -\log_2 \left(\frac{1}{2}+ \frac{1}{2} \exp\left( F\left(\frac{\theta_p+\theta_q}{2}\right) - \frac{F(\theta_p)+F(\theta_q)}{2}\right) \right) \geq 0,
$$

Assume now that $k=\frac{1}{\alpha}\geq 2$ is an arbitrary integer, and let us apply the binomial expansion for $P_\alpha(p_{\theta_1},p_{\theta_2})$ in the spirit of~\cite{fdivchiorder-2013,MinkowskiDiv-2019}:

\begin{eqnarray}
\int_\calX P_\alpha(p(x),q(x)) \dmu(x) &=& \int_\calX   \left( \frac{ p_{\theta_1}(x)^{\frac{1}{k}}+ p_{\theta_2}(x)^{\frac{1}{k}}}{2}\right)^{k} \dmu(x),\\
&=& \frac{1}{2^k} \sum_{i=0}^k \binom{k}{i} \int_\calX  \left( p_{\theta_1}(x)^{\frac{1}{k}} \right)^i
 \left( p_{\theta_2}(x)^{\frac{1}{k}}\right)^{k-i} \dmu(x).
\end{eqnarray}

Let $I_{k,i}(\theta_1,\theta_2):=\int_\calX  \left(p_{\theta_1}(x)^{\frac{1}{k}}\right)^i \left(p_{\theta_2}(x)^{\frac{1}{k}}\right)^{k-i} \dmu(x)$.
Since $\frac{i}{k}\theta_1+\frac{k-i}{k}\theta_2=\theta_2+\frac{i}{k}(\theta_1-\theta_2)\in\Theta$ for $i\in\{0,\ldots, k\}$,  we get by following the calculation steps in~\cite{BR-2011}:
$$
I_{k,i}(\theta_1,\theta_2):=\exp\left(F\left(\frac{i}{k}\theta_1+\left(1-\frac{i}{k}\right)\theta_2\right)-\left(\frac{i}{k}F(\theta_1)+\left(1-\frac{i}{k}\right)F(\theta_2)\right)\right)<\infty.
$$
Notice that $I_{2,1}=C_{\Bhat}[p_{\theta_1},p_{\theta_2}]$, and $I_{k,0}=I_{k,k}=1$.

Thus we get the following closed-form formula:
\begin{eqnarray}
R_\alpha(p_{\theta_1},p_{\theta_2})&=&-\frac{1}{k-1}\log_2 \left( \frac{1}{2^k} \sum_{i=0}^k \binom{k}{i}  \exp\left(F\left(\frac{i}{k}\theta_1+\left(1-\frac{i}{k}\right)\theta_2\right)-\left(\frac{i}{k}F(\theta_1)+\left(1-\frac{i}{k}\right)F(\theta_2)\right)\right) \right).
\end{eqnarray}
\end{proof}

This closed-form formula applies in particular to the family $\{\calN(\mu,\Sigma)\}$ of (multivariate) normal distributions:
In this case, the natural parameters $\theta$ are expressed using both a {\em vector parameter} component $v$ and a {\em matrix parameter} component $M$:
\begin{equation}
\theta=(v,M)=\left(\Sigma^{-1}m,-\frac{1}{2}\Sigma^{-1}\right),
\end{equation}
and the cumulant function is:
\begin{equation}
F_{\calN}(\theta)= \frac{d}{2}\log\pi -\frac{1}{2}\log |-2M|-\frac{1}{4} v^\top M^{-1} v,
\end{equation}
where $|\cdot|$ denotes the matrix determinant.

In general, the optimal density $c^*_\alpha=\arg\min_{c\in\calD} R_\alpha(\calP,c)$ yielding the information radius $R_\alpha(\calP)$ can be 
interpreted as a {\em generalized centroid} (extending the notion of Fr\'echet means~\cite{Frechet-1948}) with respect to $(M^R_\alpha,D_\alpha^R)$, where a {\em $(M,D)$-centroid} is defined by:

\begin{Definition}[$(M,D)$-centroid]\label{def:centroid}
Let $\calP=\{(w_1,p_1),\ldots,(w_n,p_n)\}$ be a normalized weighted point set, $M$ a mean, and $D$ a distance.
Then the $(M,D)$-centroid is defined as $c_{M,D}(\calP)=\arg\min_c M(D(p_1:c),\ldots,D(p_n:c);w_1,\ldots,w_n)$.
\end{Definition}

When all the densities $p_i$'s belong to a same exponential family~\cite{EF-2014} with cumulant function $F$ (i.e., $p_i=p_{\theta_i}$), we have 
$D_\KL[p_\theta:p_{\theta_i}]=B_F(\theta_i:\theta)$ where $B_F$ denotes the Bregman divergence~\cite{BregmanKmeans-2005}:
\begin{equation}
B_F(\theta:\theta') := F(\theta)-F(\theta')-(\theta-\theta')^\top \nabla F(\theta').
\end{equation}
Let $\calV=\{(w_1,\theta_1),\ldots,(w_n,\theta_n)\}$ be the parameter set corresponding to $\calP$. 
Define 
\begin{equation}
R_F(\calV,\theta) := \sum_{i=1}^n w_i B_F(\theta_i:\theta).
\end{equation}
Then we have the equivalent decomposition of Proposition~\ref{prop:decomp}:
\begin{equation}
R_F(\calV,\theta)-R_F(\calV,\theta^*) = B_F(\theta^*:\theta),
\end{equation}
with $\theta^*=\bar{\theta}:=\sum_{i=1}^n w_i\theta_i$.
(This decomposition is used to prove Proposition~1 of~\cite{BregmanKmeans-2005}.)
The quantity $R_F(\calV)=R_F(\calV,\theta^*)$ was termed the {\em Bregman information}~\cite{BregmanKmeans-2005,SBD-2009}.
 $R_F(\calV)$ could also be called {\em Bregman information radius} according to Sibson.
Since $R_F(\calV)=\sum_{i=1}^n w_i D_\KL[p_{\bar{\theta}}:p_{\theta_i}]$, we can interpret the Bregman information as a Sibson's information radius for densities of an exponential family with respect to the arithmetic mean $M_1^R=A$ and the {\em reverse Kullback-Leibler divergence}: $D_\KL^*[p:q]:=D_\KL[q:p]$.
This observation yields us to the JS-symmetrization of distances based on generalized information radii in \S~\ref{sec:irgen}.

Sibson proved that the information radii of any order are all upper bounded (Theorem 2.8 and Theorem 2.9 of~\cite{Sibson-1969}) as follows:  
\begin{eqnarray}
R_1(\calP) &\leq& \sum_i w_i \log_2  \frac{1}{w_j}\leq \log_2 n<\infty,\label{eq:sub1}\\
R_\alpha(\calP) &\leq&  \frac{\alpha}{\alpha-1} \log_2 \left(\sum_i w_i^{\frac{1}{\alpha}}\right)  \leq \log_2 n<\infty, \quad \alpha\in(0,1)\cup(1,\infty)\label{eq:sub2}\\
R_\infty(\calP)  &\leq&  \log_2 n<\infty.\label{eq:sub3}
\end{eqnarray}

We interpret Sibson's upper bounds of Eq.~\ref{eq:sub1}-Eq.~\ref{eq:sub3} as follows:

\begin{Proposition}[Information radius upper bound]\label{prop:IRUB}
The information radius of order $\alpha$ of a weighted set of distributions is upper bounded by the discrete R\'enyi entropy of order $\frac{1}{\alpha}$ of the weight distribution: $R_\alpha(\calP)\leq H_{\frac{1}{\alpha}}^R[w]$ where $H_\alpha^R[w] := \frac{1}{1-\alpha}\log\left( \sum_i w_i^\alpha \right)$.
\end{Proposition}

\section{JS-symmetrization of distances based on generalized information radius}\label{sec:irgen}

Let us give the following definitions generalizing the information radius (i.e., Jensen-Shannon symmetrization of the distance when $|\calP|=2$) and the ordinary Jensen-Shannon divergence:

\begin{Definition}[$(M,D)$-information radius]\label{def:genIR}
Let $M$ be a weighted mean and $D$ a  distance. 
Then the  generalized information radius for a weighted set of points (e.g., vectors or densities) $(w_1,p_1),\ldots, (w_n,p_n)$ is:
$$
R_{M,D}(\calP) := \min_{c\in\calD} M\left(D[p_1:c],\ldots, D[p_n:c]; w_1,\ldots, w_n\right).
$$
\end{Definition}

We also define the $(M,D)$-centroid as follows:

\begin{Definition}[$(M,D)$-centroid]\label{def:gencentroid}
Let $M$ be a weighted mean and $D$ a statistical distance. 
Then the centroid for a weighted set of densities $(w_1,p_1),\ldots, (w_n,p_n)$ with respect to $(M,D)$ is:
$$
c_{M,D}(\calP) := \arg\min_{c\in\calD} M\left(D[p_1:c],\ldots, D[p_n:c]; w_1,\ldots, w_n\right).
$$
\end{Definition}
When $M=A$, we recover the notion of Fr\'echet mean~\cite{Frechet-1948}.
Notice that although the minimum $R_{M,D}(\calP)$ is unique, there may potentially exists several generalized centroids $c_{M,D}(\calP)$ depending on 
$(M,D)$.

The generalized information radius can be interepreted as a diversity index or an $n$-point distance.
When $n=2$, we get the following ($2$-point) distances which are considered as a generalization of the Jensen-Shannon divergence or Jensen-Shannon symmetrization:

\begin{Definition}[$M$-vJS symmetrization of $D$]\label{def:JSsym}
Let $M$ be a mean and $D$ a statistical distance. 
Then the variational  Jensen-Shannon symmetrization of $D$ is defined by the formula of a generalized information radius:
$$
D^\vJS_{M}[p:q] := \min_{c\in\calD} M\left(D[p:c], D[q:c]\right).
$$
\end{Definition}

We use the acronym $\vJS$ to distinguish it with the JS-symmetrization reported in~\cite{JSsym-2019}:
$$
D^\JS_{M}[p:q] = D^\JS_{M,A}[p:q] := \frac{1}{2}\left(D\left[p:(pq)^M_{\frac{1}{2}}\right] + D\left[q:(pq)^M_{\frac{1}{2}}\right]\right).
$$

We recover Sibson's information radius   $R_\alpha[p:q]$ induced by two densities $p$ and $q$ from Definition~\ref{def:JSsym} as 
the {\em $M_\alpha^R$-vJS symmetrization of the R\'enyi divergence $D_\alpha^R$}. 
We have ${B_F}^\vJS_A$ which is the Bregman information~\cite{BregmanKmeans-2005}.
Notice that we may skew these generalized JSDs by taking weighted mean $M_\beta$ instead of $M$ for $\beta\in(0,1)$ yielding to the general definition: 

\begin{Definition}[Skew $M_\beta$-vJS symmetrization of $D$]\label{def:skewJSsym}
Let $M_\beta$ be a weighted mean and $D$ a statistical distance. 
Then the variational  skewed Jensen-Shannon symmetrization of $D$ is defined by the formula of a generalized information radius:
$$
\boxed{D^\vJS_{M_\beta}[p:q] := \min_{c\in\calD} M_\beta\left(D[p:c], D[q:c]\right)}
$$
\end{Definition}

Notice that this definition is implicit and can be made explicit when the centroid $c^*(p,q)$ is unique:
\begin{equation}
D^\vJS_{M_\beta}[p:q] =M_\beta\left(D[p:c^*(p,q)], D[q:c^*(p,q)]\right.
\end{equation}

In particular, when $D=D_\KL$, the KLD, we obtain generalized skewed Jensen-Shannon divergences:

\begin{Definition}[Skewed $M_\beta$-vJS divergence]\label{def:genJSD}
Let $M_\beta$ be a  weighted mean for $\beta\in(0,1)$. 
Then the $M$-vJS divergence is defined by the variational formula:
$$
D_\vJS^{M_\beta}[p:q] := \min_{c\in\calD} M_\beta\left(D_\KL[p:c], D_\KL[q:c]\right).
$$
\end{Definition}

Amari~\cite{Amari-2007} obtained the {\em $(A,D_\alpha)$-information radius} and its corresponding unique centroid  for $D_\alpha$, the $\alpha$-divergence of information geometry~\cite{IG-2016}.

Brekelmans et al.~\cite{LikelihoodRatioEF-2020} studied the geometric path $(p_1p_2)^G_\beta(x) \propto p_1^{1-\beta}(x)p_2^\beta(x)$ between two distributions $p_1$ and $p_2$ of $\calD$ where $G_\beta(a,b)=a^{1-\beta}b^\beta$ (with $a,b>0$) is the weighted geometric mean.
They proved the variational formula:
\begin{equation}
(p_1p_2)^G_\beta = \min_{c\in\calD} (1-\beta)D_\KL[c:p_1]+\beta D_\KL[c:p_2].
\end{equation}
That is, $(p_1p_2)^G_\beta$ is a $G_\beta$-$D_\KL^*$ centroid, where $D_\KL^*$ is the reverse KLD.
The corresponding $(G_\beta,D_\KL^*)$-vJSD is studied is~\cite{JSsym-2019} and used in deep learning in~\cite{VIGJSD-2020}.

It is interesting to study the link between $(M_\beta,D)$-variational Jensen-Shannon symmetrization of $D$ and the 
$(M_\alpha',N_\beta')$-JS symmetrization of of~\cite{JSsym-2019}. In particular the link between $M_\beta$ for averaging in the minimization and $M_\alpha'$ the mean for generating abstract mixtures.

More generally, Brekelmans et al.~\cite{AISqpath-2020} considered the $\alpha$-divergences  extended to positive measures  (i.e., a separable divergence built as the different between a weighted arithmetic mean and a geometric mean~\cite{nielsen2020generalization}):
\begin{equation}
D_\alpha^e[p:q] := \frac{4}{1-\alpha^2} \int_\calX \left(\frac{1-\alpha}{2} p(x)+\frac{1+\alpha}{2} q(x)
- p^{{\frac{1-\alpha}{2}}}(x)q^{{\frac{1+\alpha}{2}}}(x)
\right) \dmu(x)
\end{equation}
and proved that
\begin{equation}
c^*_\beta = \arg\min_{c\in\calD} \{(1-\beta)D_\alpha^e[p_1:c]+\beta D_\alpha^e[p_2:c]\}
\end{equation}
is a density of a {\em likelihood ratio $q$-exponential family}: $c^*_\beta= \frac{p_1(x)}{Z_{\beta,q}} \exp_q(\beta \log_q \frac{p_2(x)}{p_1(x)})$ for $q=\frac{1+\alpha}{2}$.
That is, $c^*_\beta$ is the  $(A_\beta,D_\alpha^e)$- generalized centroid, and the corresponding information radius is the variational JS symmetrization:
\begin{equation}
{D_\alpha^e}^\vJS[p_1:p_2]=(1-\beta)D_\alpha^e[p_1:c^*_\beta]+\beta D_\alpha^e[p_2:c^*_\beta]
\end{equation}
The $q$-divergence~\cite{AmariOhara-2011} $D_q$ between two densities of a $q$-exponential family amounts to a Bregman divergence~\cite{AmariOhara-2011,IG-2016}.
Thus $D_q^\vJS$ for $M=A$ is a generalized information radius which amounts to a Bregman information.

For the case $\alpha=\infty$ in Sibson's information radius, we find that the information radius is related to the total variation:
\begin{Proposition}[Lemma 2.4~\cite{Sibson-1969}]:
\begin{equation}
D^{\vJS,R}_\infty[p:q]=\log_2 (1+D_\TV[p:q]),
\end{equation}
where $D_\TV$ denotes the   total variation 
\begin{equation}
D_\TV[p:q]=\frac{1}{2}\int_\calX |p(x)-q(x)| \dmu(x).
\end{equation}
\end{Proposition}

\begin{proof}
Since $\max\{p(x),q(x)\}=\frac{p(x)+q(x)}{2}+\frac{1}{2}|q(x)-p(x)|$, it follows that $\int_\calX \max\{p(x),q(x)\}\dmu(x)=1+D_\TV[p:q]$.
From Theorem~\ref{thm:Sibson}, we have $R_\infty(\{(\frac{1}{2},p),(\frac{1}{2},q))=\log_2 \int_\calX \max\{p(x),q(x)\}\dmu(x)$ and
therefore $R_\infty(\{(\frac{1}{2},p),(\frac{1}{2},q))=\log_2\left(1+D_\TV[p:q]\right)$.
\end{proof}

\begin{Remark}
Consider the metric topology $\tau$ induced by the total variation distance $D_\TV$.
Let $B_\alpha(c,r):=\{ p\ :\ R_\alpha(c,p)<r\}$ denote the open ball centered at $c$ and of radius $r$ with respect to the information radius $R_\alpha$.
Then the set of all open balls $\{B_\alpha(c,r)\ :\ c\in\calD, \epsilon>0\}$ form a {\em basis} of $\tau$, see~\cite{Csiszar-1967}. 
\end{Remark}

Notice that when $M=M_g$ is a quasi-arithmetic mean, we may consider the divergence $D_g[p:q]=g^{-1}(D[p:q))$ so that the centroid of the
$(M_g,D_g)$-JS symmetrization is:
\begin{equation}
\arg\min_{c} g^{-1}\left( \sum_{i=1}^n w_i D[p_i:c] \right) \equiv \arg\min_{c}  \sum_{i=1}^n w_i D[p_i:c]. 
\end{equation}

The generalized $\alpha$-skewed Bhattacharyya divergence~\cite{GenBhat-2014} can also be considered with respect to a weighted mean $M_\alpha$:
$$
D_{\Bhat,M_\alpha}[p:q]=-\log\int_\calX M_\alpha(p(x),q(x)) \dmu(x).
$$
In particular, when $M_\alpha$ is a quasi-arithmetic weighted mean induced by a strictly continuous and monotone function $g$, we have
$$
D_{\Bhat,g,\alpha}[p:q]:=-\log\int_\calX M_g(p(x),q(x);\alpha) \dmu(x)=:D_{\Bhat,({M_g})_\alpha}[p:q].
$$
Since $\min\{p(x),q(x)\}\leq M_g(p(x),q(x);\alpha)\leq \min\{p(x),q(x)\}$, $\min\{a,b\}=\frac{a+b}{2}-\frac{|b-a|}{2}$ and $\max\{a,b\}=\frac{a+b}{2}+\frac{|b-a|}{2}$, we deduce that we have:
\begin{equation}
 0\leq \leq 1-D_\TV[p,q] \int_\calX M_g(p(x),q(x);\alpha) \dmu(x)\leq 1+D_\TV[p,q]\leq 2.
\end{equation}

The information radius of Sibson for $\alpha\in(0,1)\cup(1,\infty)$ may be interpreted as generalized scaled $\alpha$-skewed Bhattacharyya divergences with respect to the power means $P_\alpha$ since we have
$R_\alpha(p,q)=\frac{\alpha}{\alpha-1} \log_2 \int_\calX P_\alpha(p(x),q(x);\alpha) \dmu(x)=\frac{\alpha}{1-\alpha}D_{\Bhat,P_\alpha}[p:q]$.

\section{Relative information radius and relative Jensen-Shannon divergences}\label{sec:iref}

\subsection{Relative information radius}\label{sec:relrad}

In this section, instead of considering the full space of densities $\calD$ on $(\calX,\calF,\mu)$ for performing the variational optimization of the information radius, 
we  rather consider a subfamily of (parametric) densities $\calR\subset\calD$.
Then we define accordingly the {\em $\calR$-relative Jensen-Shannon divergence} ($\calR$-JSD for short) as

\begin{equation}
D_\vJS^\calR[p:q] := \min_{c\in\calR} D_\KL[p:c]+D_\KL[q:c].
\end{equation}

In particular, Sibson~\cite{Sibson-1969} considered the {\em normal information radius} 
with $\calR=\{\calN(\mu,\Sigma)\ :\ (\mu,\Sigma)\in \bbR^d\times\bbP_{++}^d\}$, 
where $\bbP_{++}^d$ denotes the cone of $d\times d$ positive-definite matrices (positive-definite covariance matrices of Gaussian distributions).
More generally, we may consider any exponential family $\calE$~\cite{EF-2014}.

\begin{Definition}[Relative $(\calR,M)$-JS symmetrization of $D$]\label{def:relJSsym}
Let $M$ be a mean and $D$ a statistical distance. Then
$$
D^\vJS_{M,\calR}[p:q] := \min_{c\in\calR} M\left(D[p:c], D[q:c]\right).
$$
\end{Definition}
We obtain relative Jensen-Shannon divergences when $D=D_\KL$.

Grosse et al.~\cite{Grosse-2013} considered geometric and moment average paths for annealing.
They proved that when $p_1=p_{\theta_1}$ and $p_2=p_{\theta_2}$ belong to an exponential family~\cite{EF-2014} $\calE_F$ with
cumulant function $F$, we have
\begin{equation}\label{eq:GA}
(p_1p_2)_\beta^G = \frac{p_1(x)^{1-\beta}p_2(x)^\beta}{\int p_1(x)^{1-\beta}p_2(x)^\beta\dmu(x)} 
= \arg\min_{c\in \calE_F} (1-\beta)D_\KL[c:p_1]+\beta D_\KL[c:p_2]
\end{equation}
and
\begin{equation}\label{eq:MA}
p_{\bar\eta}=\arg\min_{c\in \calE_F} (1-\beta)D_\KL[p_1:c]+\beta D_\KL[c:p_2],
\end{equation}
where $\bar\eta=(1-\beta)\eta_1+\beta\eta_2$, $\eta_i=E_{p_{\theta_i}}[t(x)]$ (this is not an arithmetic mixture but an exponential family density which moment parameter which is a mixture of the parameters). 

The corresponding minima can be interpreted as relative skewed Jensen-Shannon symmetrization for the reverse KLD $D_\KL^*$ (Eq.~\ref{eq:GA}) and the relative skewed Jensen-Shannon divergence (Eq.~\ref{eq:MA}):
\begin{eqnarray}
{D_\KL^*}^\vJS_{A_\beta,\calE_F}[p_1:p_2] &=&  \min_{c\in \calE_F} (1-\beta)D_\KL^*[p_1:c]+\beta D_\KL^*[p_2:c],\\
D^\vJS_{A_\beta,\calE_F}[p_1:p_2] &=& \min_{c\in \calE_F} (1-\beta)D_\KL[c:p_1]+\beta D_\KL[c:p_2],
\end{eqnarray}
where $A_\beta(a,b):=(1-\beta)a+\beta b$ is the weighted arithmetic mean for $\beta\in(0,1)$.

Notice that when $p=q$, we have $D^\vJS_{M,\calR}[p:p]=\min_{c\in\calR} D[p:c]$ which is the {\em information projection}~\cite{InfProj-2018} with respect to $D$ of density $q$ to the submanifold $\calR$. Thus when $p\not\in\calR$, we have $D^\vJS_{M,\calR}[p:p]>0$, i.e., the relative JSDs are not proper divergences since a proper divergence ensures that $D[p:q]\geq 0$ with equality iff $p=q$.
Figure~\ref{fig:relJS} illustrates the main cases of the relative Jensen-Shannon divergenc between $p$ and $q$:
Either $p$ and $q$ are both inside or outside $\calR$, or one point is inside $\calR$ while the other point is outside $\calR$.
When $p=q$, we get an information projection when both points are outside $\calR$, and 
$D_\vJS^\calR[p:p]=0$ when $p\in\calR$. When $p,q\in\calR$ with $p\not=q$, the value $D_\vJS^\calR[p:q]$ corresponds to the information radius (and the arg min to the right-sided Kullback-Leibler centroid).

\begin{figure}
\centering
\includegraphics[width=\textwidth]{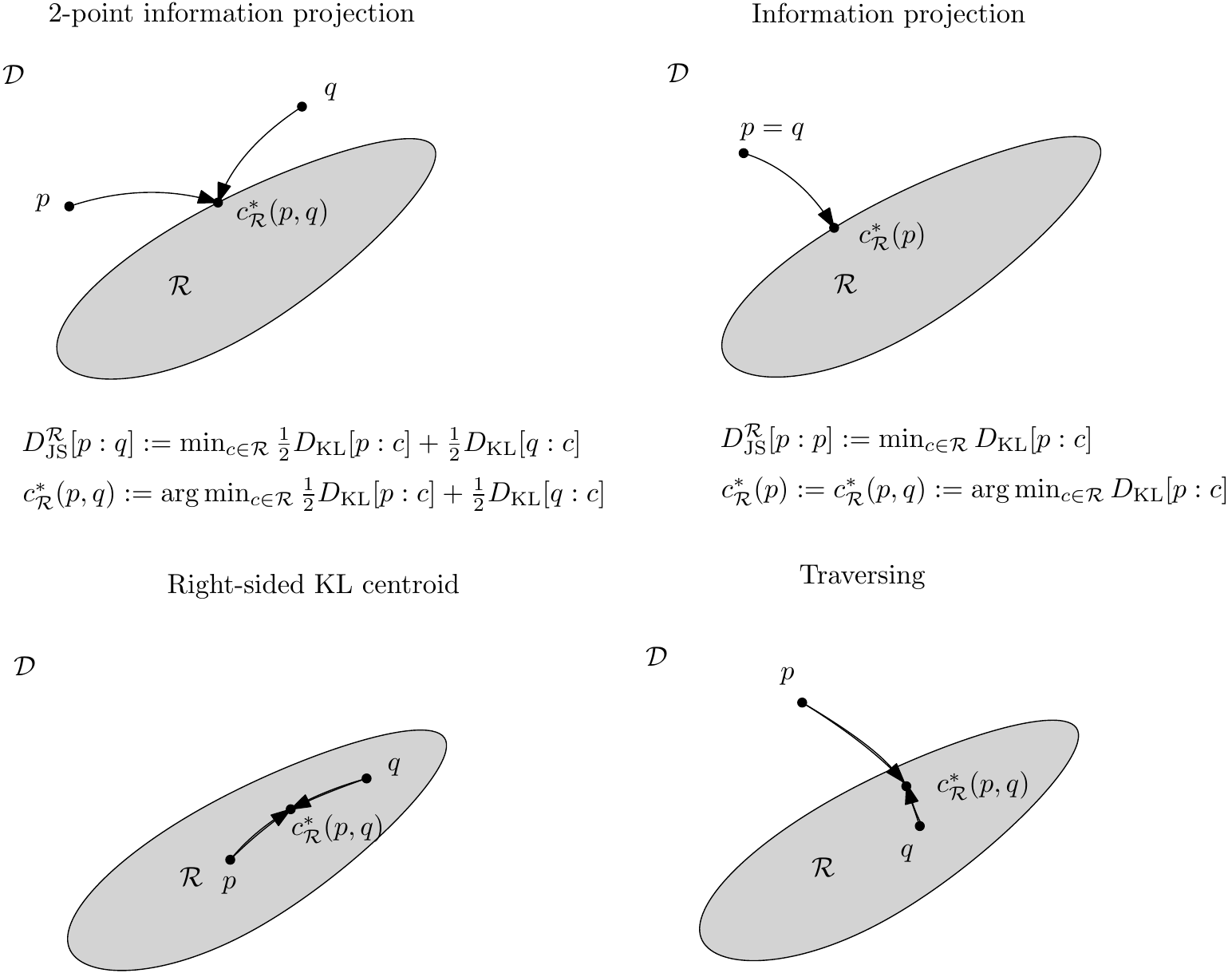}

\caption{Illustrating several cases of the relative Jensen-Shannon divergence based on whether $p\in\calR$ and $q\in\calR$ or not.\label{fig:relJS}}
\end{figure}

\subsection{Relative Jensen-Shannon divergences: Applications to density clustering and quantization}
\label{sec:relJSD}\label{sec:clustering}

Let $D_\KL[p:q_\theta]$ be the Kullback-Leibler divergence between an {\em arbitrary} density $p$ 
and a density $q_\theta$ of an exponential family $\calQ=\{q_\theta\ :\ \theta\in\Theta\}$. 
Let us express canonically~\cite{EF-2009,EF-2014} the density $q_\theta(x)$ as 
$$
q_\theta(x)=\exp(\theta^\top t_\calQ(x)-F_\calQ(\theta)+k_\calQ(x)),
$$ 
where $t_\calQ(x)$ denotes the sufficient statistics, $k_\calQ(x)$ is an auxiliary carrier measure term (e.g., $k(x)=0$ for the Gaussian family 
and $k(x)=\log(x)$ for the Rayleigh family~\cite{EF-2009}),  and 
$F_\calQ(\theta)$ the cumulant function.
Assume we both know in closed-form the following quantities:
\begin{itemize}
\item $m_p:=E_p[t_\calQ(x)]=\int p(x)t_\calQ(x) \dmu(x)$ and 
\item the Shannon entropy $h[p]=-\int p(x)\log p(x) \dmu(x)$ of $p$.
\end{itemize}
Then we can express the KLD using a {\em semi-closed-form formula}.

\begin{Proposition}\label{eq:kldmef}
Let $q_\theta\in\calQ$ be a density of an exponential family 
and $p$ an arbitrary density with $m_p=E_p[t_\calQ(x)]$.
Then the Kullback-Leibler divergence between $p$ and $q_\theta$ is expressed  as: 
\begin{equation}\label{eq:scfef}
D_\KL[p:q_\theta] = F_\calQ(\theta)-m_p^\top\theta-E_p[k_\calQ(x)]-h[p],
\end{equation}
where $h[p:q_\theta]=F_\calQ(\theta)-m_p^\top\theta-E_p[k_\calQ(x)]$ is the cross-entropy between $p$ and $q_\theta$.
\end{Proposition}

\begin{proof}
The proof is straightforward since $\log q_\theta(x)=\theta^\top t_\calQ(x)-F_\calQ(\theta)+k_\calQ(x)$.
Therefore, we have:
\begin{eqnarray}
D_\KL[p:q_\theta] &=& h[p:q_\theta]-h[p],\\
&=&-\int_\calX p(x)\log q_\theta(x) \dmu(x) - h[p],\\
&=& F_\calQ(\theta)-m_p^\top\theta-E_p[k_\calQ(x)]-h[p]. 
\end{eqnarray}
\end{proof}

\begin{Example}
For example, when $q_\theta=q_{\mu,\Sigma}$ is the density of a Gaussian distribution $\mathcal{N}(\mu,\Sigma)$ (with $k_\calN(x)=0$), we have
\begin{equation}
D_\KL[p:q_{\mu,\Sigma}]=\frac{1}{2}\left(\log |2\pi\Sigma|+
(\mu-m)^\top \Sigma^{-1}(\mu-m) + \tr(\Sigma^{-1}S) \right)-h[p],
\end{equation}
where $m=\mu(p)=E_p[x]$ and  $S=\mathrm{Cov}(p):=E_p\left[XX^\top\right]-E_p[X]E_p[X]^\top$.  
\end{Example}

The formula of Proposition~\ref{eq:kldmef} is said in semi-closed-form because it relies on knowing both the entropy $h$ of $p$ and the  sufficient statistic moments $E_p[t_\calQ(x)]$. 
Yet semi-closed formula may prove useful in practice:
For example, we can answer the comparison predicate ``$D_\KL[p:q_{\theta_1}]\geq D_\KL[p:q_{\theta_2}]$ or not?'' by checking whether
$F_\calQ(\theta_1)-F_\calQ(\theta_2)-m_p^\top(\theta_1-\theta_2)\geq 0$ or not (i.e., the terms $-E_p[k_\calQ(x)]-h[p]$ cancel out).
This is a closed-form predicate although $D_\KL$ is known only in semi-closed-form.
This KLD comparison predicate shall be used later on when clustering densities with respect to centroids in~\S\ref{sec:clustering}. 

\begin{Remark}
Note that when $Y=f(X)$ for an invertible and differentiable transformation $f$ then we have $h[Y]=h[X]+E_X[\log |J_f(X)|]$ where $J_f$ denotes the Jacobian matrix.
For example, when $Y=f(X)=AX$, we have $h[Y]=h[X]+\log |A|$.
\end{Remark}

When $p$ belongs to an exponential family $\calP$  ($\calP$ may be different from $\calQ$) with cumulant function $F_\calP$, sufficient statistics $t_\calP(x)$, 
auxiliary carrier term $k_\calP(x)$ and natural parameter $\theta$, 
  we have the entropy~\cite{crossentropyEF-2010} expressed as follows: 
\begin{eqnarray}
h[p] &=& F_\calP(\theta)-\theta^\top\nabla F_\calP(\theta)-E_p[k_\calP(x)],\\
&=& -F^*_\calP(\eta)-E_p[k_\calP(x)],
\end{eqnarray}
where $F^*_\calP(\eta)$ is the Legendre transform of $F(\theta)$ and $\eta=\eta(\theta)=\nabla F(\theta)$ is called 
the moment parameter since $\eta(\theta)=E_p[t_\calP(x)]$~\cite{EF-2009,EF-2014}.

It follows the following proposition refining Proposition~\ref{eq:kldmef} when $p=p_\theta\in\calP$:

\begin{Proposition}\label{eq:klddiffef}
Let $p_{\theta}$ be a density of an exponential family $\calP$ and
 $q_{\theta'}$ be a density of an exponential family $\calQ$. 
Then the Kullback-Leibler divergence between $p_{\theta}$ and $q_{\theta'}$ is expressed  as: 
\begin{equation}\label{eq:htimesk}
D_\KL[p_{\theta}:q_{\theta'}] = 
F_\calQ(\theta')+F_\calP^*(\eta)-E_{p_{\theta}}[t_\calQ(x)]^\top\theta' +E_{p_{\theta}}[k_\calP(x)-k_\calQ(x)].
\end{equation}
\end{Proposition}

\begin{proof}
We have 
\begin{eqnarray}
D_\KL[p_{\theta}:q_{\theta'}] &=& h[p_{\theta}:q_{\theta'}]-h[p_{\theta}],\\
&=&  F_\calQ(\theta')-m_{p_\theta}^\top\theta'-E_{p_{\theta}}[k_\calQ(x)] + F^*_\calP(\eta)+E_{p_{\theta}}[k_\calP(x)],\\
&=& F_\calQ(\theta') + F^*_\calP(\eta)- E_{p_{\theta}}[t_\calQ(x)]^\top \theta' +E_{p_{\theta}}[k_\calP(x)-k_\calQ(x)].
\end{eqnarray}
\end{proof}

In particular, when $p$ and $q$ belong both to the same exponential family (i.e., $\calP=\calQ$ with $k_\calP(x)=k_\calQ(x)$), we have $F(\theta):=F_\calP(\theta):=F_\calQ(\theta)$ and $
E_{p_{\theta}}[t_\calQ(x)]=\nabla F(\theta)=:\eta$, and 
$$
D_\KL[p_{\theta}:q_{\theta'}]=F(\theta')+F^*(\eta)-\theta'^\top\eta.
$$
This last equation is the Legendre-Fenchel divergence in Bregman manifolds~\cite{BregmanManifold-2021} (called dually flat spaces in information geometry~\cite{IG-2016}).
The divergence can thus be rewritten as equivalent dual Bregman divergences:
\begin{eqnarray}
D_\KL[p_{\theta}:q_{\theta'}] &=& F(\theta')-\eta^\top\theta'+F^*(\eta),\\
&=& B_F(\theta':\theta),\\
&=& B_{F^*}(\eta:\eta'),
\end{eqnarray}
where $\eta'=\nabla F(\theta')$.

\begin{Example}\label{ex:Weibull}
Let us use the formula of Eq.~\ref{eq:htimesk} to  calculate the KLD between two Weibull distributions~\cite{KLDWeibull-2013}.
A Weibull distribution of {\em shape} $\kappa>0$ and {\em scale} $\sigma>0$ has density defined on $\calX=[0,\infty)$ as follows:
$$
p^\Wei_{\kappa,\sigma}(x) := \frac{\kappa}{\sigma} \left(\frac{x}{\sigma}\right)^{\kappa-1} 
\exp\left(-\left(\frac{x}{\sigma}\right)^\kappa\right).
$$

For a fixed shape $\kappa$, the set of Weibull distributions $\{p^\Wei_{\kappa,\sigma}\ :\ \sigma>0\}$ form an exponential family with natural parameter $\theta=-\frac{1}{\sigma^\kappa}$, sufficient statistic $t_\kappa(x)=x^\kappa$, 
auxiliary carrier term $k_\kappa(x)=(\kappa-1)\log x+\log \kappa$, and cumulant function $F_\kappa(\theta)=-\log(-\theta)$ (so that 
$F_\kappa(\theta(\sigma))=F_\kappa(\sigma)=\kappa\log\sigma$):

$$
p^\Wei_{\kappa,\sigma}(x):=\exp\left(-\frac{1}{\sigma^\kappa} x^k +\log\frac{1}{\sigma^\kappa}+k(x)\right).
$$

We recover the exponential family of exponential distributions of rate parameter $\lambda=\frac{1}{\sigma}$ when $\kappa=1$:
\begin{eqnarray*}
p^\Exp_\lambda(x)&=&p^\Wei_{1,\sigma}(x)=\frac{1}{\sigma}\exp\left(-\frac{x}{\sigma}\right),\\
&=& \lambda\exp\left(-\lambda x\right),
\end{eqnarray*}
 and the exponential family of Rayleigh distributions when $\kappa=2$ with scale parameter $\sigma_\Ray=\frac{\sigma}{\sqrt{2}}$:
\begin{eqnarray*}
p^\Ray_{\sigma_\Ray}(x)&=&p^\Wei_{2,\sigma}(x)=\frac{2x}{\sigma^2}\exp\left(-\frac{x^2}{\sigma^2}\right),\\
&=&\frac{x}{\sigma_\Ray^2}\exp\left(-\frac{x^2}{2\sigma_\Ray^2}\right).
\end{eqnarray*}

Now, assume that we are given the differential entropy of the Weibull distributions~\cite{diffentropy-2013} (pp. 155-156):
$$
h\left[p^\Wei_{\kappa_1,\sigma_1}\right]=\gamma\left(1-\frac{1}{\kappa_1}\right)+\log\frac{\sigma_1}{\kappa_1}+1,
$$
where $\gamma\approx 0.5772156649$ is the Euler–Mascheroni constant, and the Weibull raw moments~\cite{diffentropy-2013} (p. 155):
$$
m=E_{p^\Wei_{\kappa_1,\sigma}}\left[x^{\kappa_2}\right] = \sigma_1^{\kappa_2} \Gamma\left(1+\frac{\kappa_2}{\kappa_1}\right),
$$
where $\Gamma(x)=\int_0^\infty t^{x-1} e^{-t} \mathrm{d}t$ is the gamma factorial function.
Since $h[p^\Wei_{\kappa,\sigma}]=F_\kappa(\theta)-\theta^\top\nabla F_\kappa(\theta)-E_{p^\Wei_{\kappa,\sigma}}[k_\kappa(x)]
=-F^*_\kappa(\eta)-E_{p^\Wei_{\kappa,\sigma}}[k_\kappa(x)]$, we deduce that
$$
E_{p^\Wei_{\kappa,\sigma}}[k_\kappa(x)]=-F_\kappa^*(\eta)-h\left[p^\Wei_{\kappa,\sigma}\right],
$$
where $F^*_\kappa(\eta)$ is the Legendre transform of $F_\kappa(\theta)$ and $\eta(\theta)=\nabla F_\kappa(\theta)=-\frac{1}{\theta}=E[t(x)]=E[x^\kappa]$.
We have $\theta(\eta)=\nabla F^*_\kappa(\eta)=-\frac{1}{\eta}$ and $F^*_\kappa(\eta)=\eta^\top\nabla F^*_\kappa(\eta)-F_\kappa(\nabla F^*_\kappa(\eta))=-1-\log\eta$.
It follows that 
$$
E_{p^\Wei_{\kappa,\sigma}}[k_\kappa(x)]=1+\log\left(\sigma\Gamma\left(1+\frac{1}{\kappa}\right)\right)-\gamma\left(1-\frac{1}{\kappa}\right)-\log\frac{\sigma}{\kappa}+1.
$$
Therefore, we deduce that the logarithmic moment of $p^\Wei_{\kappa_1,\sigma}$ is:
$$
E_{p^\Wei_{\kappa_1,\sigma}}[\log x]=-\frac{\gamma}{\kappa_1}+\log\sigma_1.
$$
This coincides with the explicit definite integral calculation reported in~~\cite{KLDWeibull-2013}.

Then we calculate the KLD between two Weibull distributions using Eq.~\ref{eq:htimesk} as follows:

\begin{eqnarray}
D_\KL\left[p^\Wei_{\kappa_1,\sigma_1}:p^\Wei_{\kappa_2,\sigma_2}\right] &=& 
F_{\kappa_2}(\theta') + F^*_{\kappa_1}(\eta)- E_{p_{\kappa_1,\sigma_1}}[x^{\kappa_2}]^\top \theta' +E_{p_{\kappa_1,\sigma_1}}[k_{\kappa_1}(x)-k_{\kappa_2}(x)]
\\
&=& \log \frac{\kappa_{1}}{\sigma_{1}^{\kappa_{1}}}-\log \frac{\kappa_{2}}{\sigma_{2}^{\kappa_{2}}}+
\left(\kappa_{1}-\kappa_{2}\right)\left[\log \sigma_{1}-\frac{\gamma}{\kappa_{1}}\right]+\left(\frac{\sigma_{1}}{\sigma_{2}}\right)^{\kappa_{2}} \Gamma\left(\frac{\kappa_{2}}{\kappa_{1}}+1\right)-1,
\end{eqnarray}
since we have the following terms:
\begin{eqnarray*}
F_{\kappa_2}(\theta') &=&  \log \sigma_2^{\kappa_2},\\
F^*_{\kappa_1}(\eta) &=&  -1-\log \sigma_1^{\kappa_1},\\
- E_{p_{\kappa_1,\sigma_1}}[x^{\kappa_2}]^\top \theta' &=& \frac{1}{\sigma_2^{\kappa_2}} \sigma_1^{\kappa_2}\Gamma\left(1+\frac{\kappa_2}{\kappa_1}\right)\\
E_{p_{\kappa_1,\sigma_1}}[k_{\kappa_1}(x)-k_{\kappa_2}(x)] &=& (\kappa_1-\kappa_2)E_{p_{\kappa_1,\sigma_1}}[\log x]+\log\frac{\kappa_1}{\kappa_2},\\
&=& \log\frac{\kappa_1}{\kappa_2} +  (\kappa_1-\kappa_2)\left(\log\sigma_1-\frac{\gamma}{\kappa_1}\right).
\end{eqnarray*}

This formula matches the formula reported in~\cite{KLDWeibull-2013}.

When $\kappa_1=\kappa_2=1$, we recover the ordinary KLD formula between two exponential distributions~\cite{EF-2009} with $\lambda_i=\frac{1}{\sigma_i}$ since $\Gamma(2)=1$:
\begin{eqnarray}
D_\KL\left[p^\Wei_{1,\sigma_1}:p^\Wei_{1,\sigma_2}\right] &=&  \log\frac{\sigma_2}{\sigma_1}+ \frac{\sigma_1}{\sigma_2}-1,\\
&=& \frac{\lambda_2}{\lambda_1}-\log\frac{\lambda_2}{\lambda_1}-1.
\end{eqnarray}

When $\kappa_1=\kappa_2=2$, we recover the ordinary KLD formula between two Rayleigh distributions~\cite{EF-2009} with 
$\sigma_\Ray=\frac{\sigma}{\sqrt{2}}$:
 \begin{eqnarray}
D_\KL\left[p^\Wei_{2,\sigma_1}:p^\Wei_{2,\sigma_2}\right] &=&  \log\left(\frac{\sigma_2^2}{\sigma_1^2}\right)+ \frac{\sigma_1^2}{\sigma_2^2}-1,\\
&=& \log\left(\frac{{\sigma_\Ray}_2^2}{{\sigma_\Ray}_1^2}\right)+ \frac{{\sigma_\Ray}_1^2}{{\sigma_\Ray}_2^2}-1.
\end{eqnarray}

\end{Example}

To find the best density $q_\theta$ approximating $p$ by minimizing $\min_{\theta} D_\KL[p:q_\theta]$, we solve $\nabla F(\theta)=\eta=m$, and therefore
$\theta=\nabla F^*(m)=(\nabla F)^{-1}(m)$ where $F^*(\eta)=E_{q_\eta}[\log q_\eta(m)]$ with $F^*$ denoting the Legendre-Fenchel convex conjugate~\cite{EF-2014}.
In particular, when $p=\sum w_i p_{\theta_i}$ is a mixture of EFs (with $m=E_p[t(x)]=\sum w_i\eta_i$ with $\eta_i=E_{p_{\theta_i}}[t(x)]$ thanks to the linearity of the expectation), then the best density of the EF simplifying $p$ is

\begin{eqnarray}
\min_\theta D_\KL[p:q_\theta] &=& \min_\theta F(\theta)-m^\top\theta,\\
&=& \min_\theta F(\theta)-\sum w_i\eta_i^\top\theta.
\end{eqnarray}

Taking the gradient with respect to $\theta$, we have $\nabla F(\theta)=\eta=\sum w_i\eta_i$.
This yields another proof without the Pythagoras theorem~\cite{Pelletier-2005,LearningMixtureKde-2013}.

\begin{Proposition}
Let $m(x)=\sum w_i p_{\theta_i}(x)$ be a mixture with components belonging to an exponential family with cumulant function $F$.
Then $\theta^*=\arg_\theta \min_{\theta} D_\KL[p:q_\theta]$ is $\nabla F^*(\sum_{i=1}^n w_i\eta_i)$ where the $\eta_i=\nabla F(\theta_i)$ are the moment parameters of the mixture components.
\end{Proposition}

Consider the following two problems:

\begin{Problem}[Density clustering]\label{pb:dc}
Given a set of $n$ weighted densities $(w_1,p_1), \ldots, (w_n,p_n)$, partition them into $k$ clusters $\calC_1,\ldots,\calC_k$ in order to minimize the $k$-centroid objective function with respect to a statistical divergence $D$:
$\sum_{i=1}^n w_i \min_{l\in\{1,\ldots,k\}}  D[p_i:c_l]$, where $c_l$ denotes the centroid of cluster $\calC_l$ for $l\in\{1,\ldots,k\}$.
\end{Problem}

For example, when all densities $p_i$'s are isotropic Gaussians, we recover the $k$-means objective function~\cite{kmeans-1982}.

\begin{Problem}[Mixture component quantization]\label{pb:mq}
Given a statistical mixture $m(x)=\sum_{i=1}^n w_ip_i(x)$, quantize the mixture components into $k$ densities $q_1,\ldots, q_k$ in order to minimize
$\sum_i w_i \min_{l\in\{1,\ldots,k\}}   D[p_i:q_l]$.
\end{Problem}

Notice that in Problem~\ref{pb:dc}, the input densities $p_i$'s may be mixtures, i.e., $p_i(x)=\sum_{j=1}^{n_i} w_{i,j}p_{i,j}(x)$.
Using the relative information radius, we can cluster a set of distributions (potentially mixtures) into an exponential family mixture, or quantize an exponential family mixture.
Indeed, we can implement an extension of $k$-means~\cite{kmeans-1982} with $k$-centers $q_{\theta_i}$, to assign density $p_i$ to cluster $C_j$ (with center $q_j$), we need to perform basic comparison tests
$D_\KL[p_i:q_{\theta_l}]\geq D_\KL[p_i:q_{\theta_j}]$. 
Provided the cumulant $F$ of the exponential family is in closed-form, we do not need formula for the entropies $h(p_i)$.

Clustering and quantization of densities/mixtures have been widely studied in the literature, see for example~\cite{EntropicClusteringGaussians-2006,ClusteringGaussian-2008,QuantizationBregman-2010,SimplifyingMixtures-2010,ClusteringGMM-2013,MusicGMM-2015,ClusteringGaussian-2016}.

\section{Conclusion}\label{sec:concl}

\begin{figure}
\centering
\includegraphics[width=\textwidth]{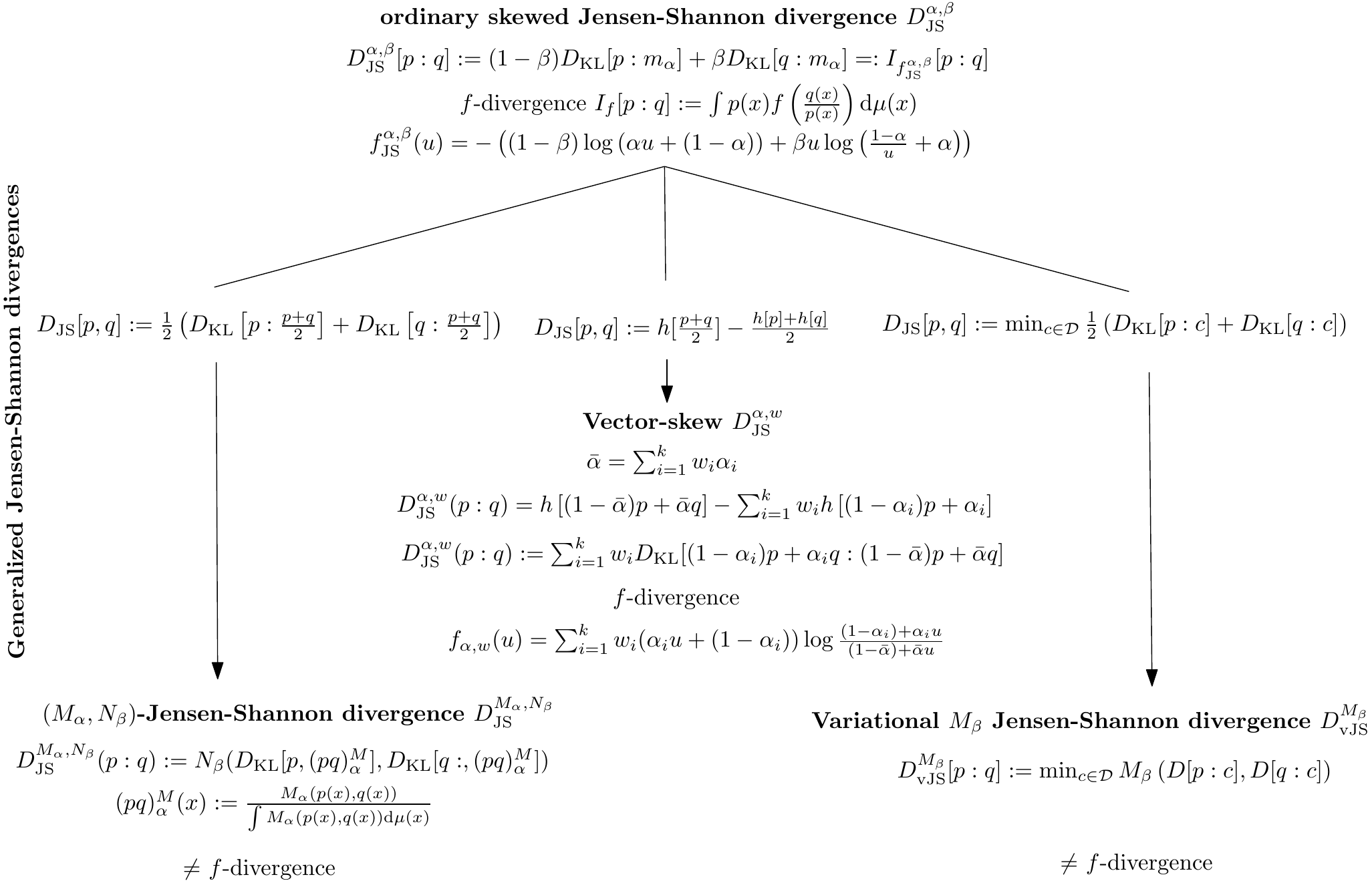}
\caption{Three equivalent expressions of the ordinary (skewed) Jensen-Shannon divergence which yield three different generalizations.\label{fig:genJSdiag}}
\end{figure}

To summarize, the ordinary Jensen-Shannon divergence has been defined in three equivalent ways in the literature:
\begin{eqnarray}
D_\JS[p,q] &:=& \min_{c\in\calD} \frac{1}{2}\left( D_\KL[p:c]+D_\KL[q:c] \right),\label{eq:js1}\\ 
&=&  \frac{1}{2}\left(D_\KL\left[p:\frac{p+q}{2}\right]+D_\KL\left[q:\frac{p+q}{2}\right]\right),\label{eq:js2}\\ 
&=&  h\left[\frac{p+q}{2}\right] - \frac{h[p]+h[q]}{2}. \label{eq:js3}
\end{eqnarray} 

The JSD Eq.~\ref{eq:js1} was studied by Sibson in 1969 within the wider scope of information radius~\cite{Sibson-1969}:
Sibson relied on the R\'enyi $\alpha$-divergences (relative R\'enyi $\alpha$-entropies~\cite{Entropy-1995})
 and recovered the ordinary Jensen-Shannon divergence as a particular case of the $\alpha$-information radius when $\alpha=1$ and $n=2$ points.

The JSD Eq.~\ref{eq:js2} was investigated by Lin~\cite{JS-1991} in 1991 with its connection to the JSD defined in Eq.~\ref{eq:js2}).
In Lin~\cite{JS-1991}, the JSD is interpreted as the arithmetic symmetrization of the $K$-divergence~\cite{nielsen2010family}.
Generalizations of the JSD based on Eq.~\ref{eq:js2} was proposed in~\cite{JSsym-2019} using a generic mean instead of the arithmetic mean.
One motivation was to obtain a closed-form formula for the geometric JSD between multivariate Gaussian distributions which relies on the geometric mixture (see~\cite{VIGJSD-2020} for a use case of that formula in deep learning).
Indeed, the ordinary JSD between Gaussians is not available in closed-form (not analytic).
However, the JSD between Cauchy distributions admit a closed-form formula~\cite{CauchyJSD-2021} despite the calculation of a definite integral of a log-sum term. Instead of using an abstract mean to define a mid-distribution of two densities, one may also consider the mid-point of a geodesic linking these two densities (the arithmetic means $\frac{p+q}{2}$ is interpreted as a geodesic midpoint). 
Recently, Li~\cite{TransportInfoBD-2021} investigated the transport Jensen-Shannon divergence as a symmetrization of the Kullback-Leibler divergence in the $L^2$-Wasserstein space. 
See Section 5.4 of~\cite{TransportInfoBD-2021} and the closed-form formula of Eq.~18 obtained for the transport Jensen-Shannon divergence between two multivariate Gaussian distributions.

Generalization of the identity between the JSD of Eq.~\ref{eq:js2} and the JSD of Eq.~\ref{eq:js3} was studied using a skewing vector in~\cite{JScentroid-2020}. 
Although the JSD is a $f$-divergence~\cite{Csiszar-1964,JScentroid-2020}, the Sibson-$M$ Jensen-Shannon symmetrization of a distance does not belong in general to the class of $f$-divergences.
The variational JSD definition of Eq.~\ref{eq:js1} is implicit while the definitions of Eq.~\ref{eq:js2} and Eq.~\ref{eq:js3} are explicit because the unique optimal centroid $c^*=\frac{p+q}{2}$ has been plugged into the objective function minimized by Eq.~\ref{eq:js1}.

In this paper, we proposed a generalization of the Jensen-Shannon divergence based on the variational definition of the ordinary Jensen-Shannon divergence based on the variational JSD definition of Eq.~\ref{eq:js1}: $D_\vJS[p:q]=\min_c \frac{1}{2}(D_\KL[p:c]+D_\KL[q:c])$.
We introduced the Jensen-Shannon symmetrization of an arbitrary divergence $D$ by considering a generalization of the information radius with respect to an abstract weighted mean $M_\beta$: $D^\vJS_M[p:q]:=\min_c M_\beta(D[p:c],D[q:c])$.
Notice that in the variational JSD, the mean $M_\beta$ is used for averaging divergence values, while the mean $M_\alpha$ in the $(M_\alpha,N_\beta)$ JSD is used to define generic statistical mixtures.
We also consider {\em relative} variational JS symmetrization when the centroid has to belong to a prescribed family of densities.
For the case of exponential family, we showed how to compute the relative centroid in closed form, thus extending the pioneer work of Sibson who considered the relative normal centroid used to calculate the relative normal information radius. 
Figure~\ref{fig:genJSdiag} illustrates the three generalizations of the ordinary skewed Jensen-Shannon divergence.
Notice that in general, the $(M,N)$-JSDs and the variational JDSs are not $f$-divergences (except in the ordinary case).

In a similar vein, Chen et al.~\cite{Chen-BD2-2008} considered the following minimax symmetrization of the scalar Bregman divergence~\cite{Bregman-1967}:
\begin{eqnarray}\label{eq:varbd}
B^\minmax_f(p,q) &:=& \min_c \max_{\lambda\in[0,1]} \lambda B_f(p:c)+(1-\lambda) B_f(q:c),\\
&=& \max_{\lambda\in[0,1]} \lambda B_f(p:\lambda p+(1-\lambda) q)+(1-\lambda) B_f(q:\lambda p+(1-\lambda)),\\
&=& \lambda f(p)+(1-\lambda)f(q)-f(\lambda p+(1-\lambda))
\end{eqnarray}
where $B_f$ denotes the scalar Bregman divergence induced by a strictly convex and smooth function $f$:
\begin{equation}
B_f(p:q)=f(p)-f(q)-(p-q)f'(q).
\end{equation}
They proved that $\sqrt{B^\minmax_f(p,q)}$ yields a metric when $3(\log f'')''\geq ((\log f'')')^2$, and extend the definition to the vector case and conjecture that the square-root metrization still holds in the multivariate case. 
In a sense, this definition geometrically highlights the notion of radius since the minmax optimization amount to find a smallest enclosing ball enclosing~\cite{minmax-2013} the source distributions. The circumcenter also called the Chebyshev center~\cite{ChebyshevAlphaDiv-2020} is then the mid-distribution instead of the centroid for the information radius.
The term ``information radius'' is well-suited to measure the distance between two points for an arbitrary distance $D$.
Indeed, the JS-symmetrization of $D$
 is defined by  $D^\JS[p:q]:=\min_c \{\frac{1}{2}D[p:c]+\frac{1}{2}D[q:c]\}$.
When $D[p:q]=D_E[p:q]=\|p-q\|$ is the Euclidean distance, we have $c=\frac{p+q}{2}$, and 
$D[p:c]=D[q:c]=\frac{1}{2}\|p-q\|=:r$ (i.e., the radius being half of the diameter $\|p-q\|$). 
Thus $D^\JS_E[p:q]=r$, hence the term chosen by Sibson~\cite{Sibson-1969} for $D^\JS$: information radius.
Besides providing another viewpoint, variational definitions of divergences are proven useful in practice (e.g., for estimation).
For example, a variational definition of the R\'enyi divergence generalizing the Donsker-Varadhan variational formula of the KLD is given in~\cite{variationalRenyi-2020} which is used to estimate the R{\'e}nyi Divergences.

\vskip 0.5cm
\noindent{\bf Acknowledgments}: We warmly thank Rob Brekelmans (Information Sciences Institute, University of Southern California, USA) for discussions and feedback related to the contents of this work. The author thanks the reviewers for valuable feedback, comments, and suggestions.

\bibliographystyle{plain}
\bibliography{ClusteringQuantizingDensitiesBIB}

\end{document}